\documentclass{article}

\usepackage{arxiv}

\usepackage[utf8]{inputenc} 
\usepackage[T1]{fontenc}    
\usepackage{url}            
\usepackage{booktabs}       
\usepackage{amsfonts}       
\usepackage{nicefrac}       
\usepackage{microtype}      
\usepackage{amsmath}
\usepackage{lipsum}         
\usepackage{graphicx}
\usepackage{doi}

\usepackage[utf8]{inputenc} %
\usepackage[T1]{fontenc}    %
\usepackage{hyperref}       %
\usepackage{url}            %
\usepackage{booktabs}       %
\usepackage{amsfonts}       %
\usepackage{nicefrac}       %
\usepackage{microtype}      %
\usepackage{xcolor}         %

\usepackage{amsfonts}
\usepackage{color}

\usepackage{bbm}
\usepackage{mathtools}
\usepackage{bm}

\usepackage{amsmath}
\usepackage{amssymb}
\usepackage{amsthm}
\usepackage{stmaryrd}
\usepackage{physics}
\usepackage{graphicx}
\usepackage{caption}
\usepackage{subcaption}

\usepackage{cellspace}

\newtheorem{theorem}{Theorem}[section]

\newtheorem{lemma}{Lemma}[section]

\newtheorem{definition}{Definition}[section]
\newtheorem{proposition}{Proposition}[section]
\newtheorem{example}{Example}[section]

\def\>{\ensuremath{\rangle}}
\def\<{\ensuremath{\langle}}

\newcommand {\cH } {{\mathcal{H}}}

\renewcommand {\Tr } {{\mathrm{Tr}}}

\newcommand {\qbit}[2]{{|#1\>_{#2}\<#1|}}
\newcommand {\id}{ \mathbbm{1} }

\newcommand\cincludegraphics[2][]{\raisebox{-0.3\height}{\includegraphics[#1]{#2}}}
\newcommand{\myproof}[1]{\textit{Proof}. #1 \hfill\qedsymbol}
\newcommand{\red}[1]{#1}

\title{Mitigating Noise-Induced Gradient Vanishing in Variational Quantum Algorithm Training}

\author{%
  Anbang Wu \\
  Department of Computer Science\\
  University of California, Santa Barbara \\
  \texttt{anbang@ucsb.edu} \\
  \And
  Gushu Li \\
  Department of Electrical \& Computer Engineering\\
  University of California, Santa Barbara \\
  \texttt{gushuli@ece.ucsb.edu} \\
  \AND
  Yufei Ding \\
  Department of Computer Science\\
  University of California, Santa Barbara \\
  \texttt{yufeiding@cs.ucsb.edu} \\
  \AND
  Yuan Xie \\
  Department of Electrical \& Computer Engineering\\
  University of California, Santa Barbara \\
  \texttt{yuanxie@ucsb.edu} \\
}

\begin{document}

\maketitle

\begin{abstract}
Variational quantum algorithms are expected to demonstrate the advantage of quantum computing on near-term noisy quantum computers. 
However, training such variational quantum algorithms suffers from gradient vanishing as the size of the algorithm increases.
Previous work cannot handle the \textbf{gradient vanishing} induced by the inevitable noise effects on realistic quantum hardware.
In this paper, we propose a novel training scheme to mitigate such noise-induced gradient vanishing.
We first introduce a new cost function of which the gradients are significantly augmented by employing \textbf{traceless observables} in \textbf{truncated subspace}.
We then prove that the same minimum can be reached by optimizing the original cost function with the gradients from the new cost function.
Experiments show that our new training scheme is highly effective for major variational quantum algorithms of various tasks.
\end{abstract}

\section{Introduction}

The recent development of quantum hardware has been able to demonstrate the `quantum supremacy'~\cite{preskill2012quantum} against classical computers on a specific application~\cite{arute2019quantum}. 
However, a quantum algorithm that is of practical usage and can demonstrate the advantage of quantum computing on near-term noisy quantum hardware platforms is still unclear. 
Variational Quantum Algorithm~(VQA)~\cite{cerezo2020variational} is one of the leading candidates because it has a relatively moderate requirement for the number of qubits and the depth of computation and it has demonstrated the ability of error tolerance. 
Several VQAs have been experimentally implemented on existing quantum computers for various problems, e.g., optimization~\cite{arute2020quantum, Pagano2020QuantumAO}, chemistry simulation~\cite{Arute2020HartreeFockOA, Kandala2017HardwareefficientVQ}.

VQA can be considered as the quantum analogy of classical machine learning~\cite{cerezo2020variational}. 
It has two key components.
The first one, usually referred to as `ansatz' in the quantum computing community, is a quantum program with some tunable parameters $\bm{\theta}$. 
This program will convert an input state $\rho$ to an output state $\rho(\bm{\theta})$.
The second component is the observable $O$, a Hermitian operator in the Hilbert space of the quantum system.
The observable is constructed based on the target problem so that its smallest eigenvalue, together with the corresponding eigenstate, is the solution to the target problem.
On a quantum computer, we can measure the value of $\Tr(\rho(\bm{\theta})O)$, which is the expectation of the observable $O$ with respect to the final state of the parameterized quantum program.
VQA then employs a classical optimizer (e.g., gradient descent) to optimize the parameters $\bm{\theta}$  and minimize the cost function $\Tr(\rho(\bm{\theta})O)$, which is similar to training a model in classical machine learning.
After we reach the minimal $\Tr(\rho(\bm{\theta})O)$, $\min_\theta \{\Tr(\rho(\bm{\theta})O)\}$ becomes the smallest eigenvalue ($\rho(\bm{\theta})$ is naturally the corresponding eigenstate) and a solution to the target problem is found.

One of the critical challenges in training a VQA is the gradient vanishing, i.e., the partial derivatives of the cost function $\Tr(\rho(\bm{\theta})O)$ with respect to its parameters are exponentially small as the system size increases~\cite{mcclean2018barren,wang2020noise}. Inspired by the related techniques in classical neural network training,  several techniques~\cite{Grant2019AnIS, Verdon2019LearningTL, Skolik2020LayerwiseLF, Kubler2019AnAO, Zhang2020TowardTO, Volkoff2020LargeGV, Pesah2020AbsenceOB, Du2020QuantumCA, cerezo2020cost} have been proposed to efficiently suppress the gradient exponential decay 
in the noise-free scenario.  %

However, all these techniques mentioned above ignores the noise effects.
Different from classical computers which are usually digital with extremely small error rates, noise effects are inevitable on the analog near-term quantum devices.
A recent study points out that the noise effect itself can cause gradient vanishing by concentrating the cost function landscape around a specific value~\cite{wang2020noise}.
Such noise-induced gradient vanishing, which is not yet considered in existing gradient vanishing mitigation techniques, will make the training of VQAs extremely difficult and hinder solving practical problems with VQAs~\cite{wang2020noise}.
It is crucial to mitigate the noise-induced gradient vanishing in VQA training before we can achieve quantum advantage with VQA.

In this paper, we focus on mitigating the noise-induced gradient vanishing through a key insight that the observables in major VQAs are only active in a small subspace. 
This property comes from the symmetry in the cost function \cite{Sagastizabal2019ErrorMB, Gard2019EfficientSS, BonetMonroig2018LowcostEM} and such symmetry is shown to be general in major VQAs for realistic problems~\cite{Shaydulin2020ClassicalSA, Wang2017TheQA, Bravyi2019ObstaclesTS}.
Some noise effects, when they lead to a measurement result outside the subspace, can be eliminated to reduce the noise effects and the noise-induced gradient vanishing can be mitigated in this process.
Besides, we notice that the trivial measurement result (i.e., the measurement with respect to the identity component in the observable) will aggravate the noise-induced gradient vanishing. Thus, the noise-induced gradient vanishing can be alleviated by disabling trivial measurements.

To this end, we propose a novel training scheme with two cost functions.
The first cost function $C_1$ is the vanilla cost in a subspace while the second function $C_2$ is crafted by making the observables traceless (i.e., the observable without identity component~\cite{Nielsen2000QuantumCA}).
Such a design can remove a constant dominating term in the observable and amplify the variation with respect to the parameters. Besides, gradient induced by measurement results outside the subspace will be omitted owing to the use of subspace.
Theoretically, we prove that the gradient of $C_2$ is augmented exponentially comparing with that of the vanilla cost $C_1$.
We also show that the same minimum value can be reached when optimizing the vanilla $C_1$ with the augmented gradient from $C_2$.

We evaluate our training scheme by simulating the noisy VQA training on major VQAs and various tasks. 
The results show that the magnitudes of the gradients of $C_2$ are $58.4\times$ larger on average compared with that of $C_1$.
For combinatorial optimization tasks, 
with the subspace derived from the symmetries of graph structures, 
our training scheme can increase the probability of measuring the correct output by $87.7\%$ compared with the conventional setup with the vanilla cost function.
For chemistry simulation tasks, with the subspace derived from the particle conservation law, the fidelity of the simulated quantum state is $68.1\%$ larger than that of the conventional setup.

Our major contributions can be summarized as follows:
\begin{enumerate}
    \item We propose the first training scheme that can counter the unavoidable noise-induced gradient vanishing in VQA training on near-term quantum computers.
    \item We design a new cost function of truncated subspace with traceless observables
    and theoretically prove that the gradients can be augmented exponentially without changing the minimum value found in the optimization.
    \item Our experiments show that our new training scheme outperforms conventional gradient descent training with significantly larger gradients, faster convergence, and better VQA results.
\end{enumerate}

\subsection{Related Work}

VQA can be considered as the quantum analog of the successful neural-network-based classical machine learning methods. 
It has demonstrated potential applicability in optimization~\cite{farhi2014quantum, moll2018quantum}, classification~\cite{farhi2018classification, Schuld2020circuit}, chemistry simulation~\cite{peruzzo2014variational,kandala2017hardware}, etc.

One of the key challenges in VQA training is the gradient vanishing problem.
\cite{mcclean2018barren,sharma2020trainability} observed that gradients experience an exponential decay with respect to the qubit numbers on deep VQA circuits. 
\cite{wang2020noise} further revealed that the gradient barren plateaus can also be induced by local device noises. 
Empirical studies~\cite{Kubler2019AnAO,Xue2019EffectsOQ} provides numerical evidence of the gradient vanishing phenomenon of VQAs in both noise-free and noisy settings.

Some of the techniques that are proposed for gradient vanishing mitigation in classical neural networks can possibly be migrated to  VQA since most VQA circuits have multiple layers with tunable parameters.
These techniques include parameter initialization~\cite{Grant2019AnIS}, pre-training~\cite{Verdon2019LearningTL}, layer-wise training~\cite{Skolik2020LayerwiseLF}, changing optimizer~\cite{Kubler2019AnAO}, and variational quantum circuit architecture design~\cite{Zhang2020TowardTO, Volkoff2020LargeGV, Pesah2020AbsenceOB, Du2020QuantumCA}.
In addition, \cite{cerezo2020cost} proposed to employ local observables in the cost function for better VQA trainability. 
All these works only consider the noise-free gradient vanishing~\cite{mcclean2018barren,sharma2020trainability}.
However, the noise effects, which do not exist in classical neural networks, are inevitable in the near term.
The techniques mentioned above cannot counter the noise-induced gradient vanishing~\cite{wang2020noise}. %
In contrast, this paper tackles this
problem by suppressing the dominated trivial measurement as well as noise effects outside the active space of the optimization target.
The gradient can then be significantly augmented, leading to better VQA trainability.

\section{Preliminary}

In this section, we introduce necessary backgrounds to help understand the proposed gradient vanishing mitigation techniques.

\subsection{Quantum Computing Basics}

The basic information processing unit in quantum computing is a qubit. 
Different from a classical bit which can be either 0 or 1 exclusively, the state of a qubit, usually denoted by $\ket{\psi}$, is a unit vector in a 2-dimensional Hilbert space $\cH_2 =\{a\ket{0}+b\ket{1}\}$ where $a,b \in \mathbb{C}$ and $\abs{a}^2 + \abs{b}^2 = 1$. 
The state space of $n$ qubits is the tensor product of the state spaces of the $n$ qubits, $\bigotimes^n\cH_2 = \cH_{2^n}$, a $2^n$-dimensional Hilbert space.
Sometimes the exact state is not known but we know it can be in one of a set of pure states $\ket{\psi_i}$ with respective probability $p_i$ and $\sum_i p_i = 1$.
We define a density operator $\rho=\sum_i p_i \qbit{\psi_i}{}$ to represent such a mixed state ($\bra{\psi}$ is the complex conjugate of $\ket{\psi}$).

There are two major types of operations applied to a quantum system. The first one is the unitary transformation which can be represented as a unitary operator $U$ in the state space where $UU^{\dagger}=I$ and $I$ is the identity operator. A state vector $\ket{\psi}$ and a density operator $\rho$ will be changed to $U\ket{\psi}$ and $U\rho U^{\dagger}$, respectively, after a unitary transformation.
Measurement can extract information from a quantum system.
What can be measured from a quantum system is defined by an observable $O$ which must be a Hermitian operator (a Hermitian operator $H$ is self-adjoint and $H=H^{\dagger}$).
In practice, we can measure the expectation value of an observable $O$ and it can be calculated by $\bra{\psi}O\ket{\psi}$ or $\Tr(\rho O)$ for a state vector $\ket{\psi}$ or density operator $\rho$.
One Hermitian operator can be 
decomposed into an array of weighted Pauli strings, which is defined as follows.

\begin{definition}[Pauli representation] 
The Pauli representation of $2^n$-dimensional Hermitian matrix $H$ is defined as, $H = \sum_i \eta^i \sigma_n^i = \frac{\Tr(H)}{2^n}\sigma_n^0 + \bm{\eta}\cdot \bm{\sigma}_n$,
where $\sigma_n^0 = \id^{\otimes n}$, $\sigma_n^i \in \{ \id, \sigma_X, \sigma_Y, \sigma_Z \}^{\otimes n} \backslash \{ \id^{\otimes n} \}$. $\id$ is the 2-dimensional identity matrix; $\sigma_X$, $\sigma_Y$ and $\sigma_Z$ are Pauli matrices, $\Tr(H)$ is the trace of $H$. $\sigma_n^i$ is often called Pauli string. $2^{n}\bm{\eta} = \Tr(H\bm{\sigma}_n)$ is often called Pauli coefficient of $H$. \red{A Pauli string together with its Pauli coefficient is referred to as a Pauli term.}
\end{definition}
\red{
\begin{definition}[Rank of observable]
The rank of an observable $O$ is defined to be the number of non-zero Pauli coefficients in its Pauli representation.
\end{definition}
}

\subsection{Variational Quantum Algorithms}

Variational quantum algorithm~\cite{cerezo2020variational} is a leading candidate that can potentially demonstrate the advantage of quantum computing. 
Variational quantum circuit (VQC) is the core component of variational quantum algorithms. An n-qubit VQC $\mathcal{U}_n^L$ consists of L layers, and can be decomposed into a sequence of unitary transformations as
$ \mathcal{U}_n^L = \mathcal{U}_L(\bm{\theta}_L) \circ  \cdots \circ \mathcal{U}_2(\bm{\theta}_2)  \circ \mathcal{U}_1(\bm{\theta}_1)$,
where $\mathcal{U}_l(\bm{\theta}_l)$ is a super operator defined as 
$\mathcal{U}_l(\bm{\theta}_l)(\rho) = U_l(\bm{\theta}_l)\rho U_l(\bm{\theta}_l)^\dagger$.
And
\begin{align}\label{equ:uni_def}
U_l(\bm{\theta}_l) &= \exp\{-i\, {\textstyle \sum}_j (\theta_{lj} H_{lj} + R_{lj})\}\,.
\end{align}
Here $\bm{\theta}_l=\{\theta_{lj}\}$ are continuous parameters, both $H_{lj}$ and $R_{lj}$ are $2^n$-dimensional Hermitian matrices.

\begin{definition}[Local Pauli channel]
The Local Pauli Channel $\mathcal{P}$ defined on a single-qubit state $\rho$ can be expressed as
$\mathcal{P}(\rho) = \sum_{i} q_i \sigma_i\rho \sigma_i + (1-\sum_{i} q_i)\rho$,
where $i \in \{ X, Y, Z \}$.
\end{definition}

In this paper, we consider the noise model based on the local Pauli channel, and the noisy n-qubit VQC model expressed as 
$\widetilde{\mathcal{U}}_n^L = \widetilde{\mathcal{U}}_L(\bm{\theta}_L) \circ  \cdots \circ \widetilde{\mathcal{U}}_2(\bm{\theta}_2)  \circ \widetilde{\mathcal{U}}_1(\bm{\theta}_1)$,
where $\widetilde{\mathcal{U}}_l(\bm{\theta}_l) = \mathcal{P}^{\otimes n}\circ \mathcal{U}_l(\bm{\theta}_l) $ is the the unitary operation distorted by noise.
We refer to the state outputted by channel $\widetilde{\mathcal{U}}_l$ as $\rho_l$.
Specifically, the final output state of $\widetilde{\mathcal{U}}_n^L$ is referred to as $\rho_L$.
Quantum state $\rho_l$ is also a Hermitian matrix, and its Pauli representation can be expressed as
$
    \rho_l = \frac{1}{2^n}\Big( \id^{\otimes n} + \bm{a}^l\cdot \bm{\sigma}_n \Big)
$,
where $a_i^l = \Tr(\rho_l\, \sigma_n^i)$.

In the following, we introduce the cost function of the n-qubit local noisy circuit, and calculate its gradient with the Pauli representation language.

\begin{definition}[Cost function of the local noisy circuit] Given an observable $O$ and its eigendecomposition $O=\sum_{i}\lambda_i \qbit{i}{}$, the cost function $C_0$ is defined as
$C_0 = \Tr(\rho_L O) = \sum_{i}\lambda_i p_i$,
where $p_i$ is the probability of $\rho_L$ being in the quantum state $\qbit{i}{}$.
\end{definition}
\begin{proposition}[Gradient of cost function $C_0$~\cite{wang2020noise}] For any parameter $\theta$ in the variational quantum circuit, we have
$\partial_\theta C_0 = \bm{g}^L_\theta\cdot \bm{w}$,
where $2^n\bm{w}$ is the Pauli coefficient of observable $O$, $\bm{g}^L_\theta$ is the Pauli coefficient of $\partial_\theta \rho_L$. If without ambiguity, we will use $\bm{g}^L$ instead of $\bm{g}^L_\theta$ for simplicity.
\end{proposition}

To be convenient, we use a single parameter $q = \max\{1-2q_Y-2q_Z,1-2q_X-2q_Z,1-2q_X-2q_Y\}$ to denote the noise strength of a local Pauli channel. We then introduce the gradient vanishing problem of VQA induced by the local Pauli error channel.

\begin{lemma}[Gradient vanishing of the Pauli~coefficient~\cite{wang2020noise}]\label{lemma:grad_van_pauli}
Let $\theta_{lm}$ be the $m$-th element of parameter $\bm{\theta}_l$, $\bm{g}^L_{lm}$ be the Pauli coefficient of $\partial_{\theta_{lm}}\rho_L$. Then
\begin{align}
    \norm{\bm{g}_{lm}^L}_\infty \le \norm{\bm{g}_{lm}^L}_2 \le \sqrt{2^{n+1}-2}N_{lm}\norm{\bm{\eta}_{lm}}_\infty q^{L+1}\,,
\end{align}
where $2^n\bm{\eta}_{lm}$ is the Pauli coefficient of $H_{lm}$ defined in Equ~(\ref{equ:uni_def}) and $N_{lm} $ is the number of non-zero elements of $\bm{\eta}_{lm}$.
\end{lemma}

The following lemma reveals the effect of the local noisy Pauli channel on the Pauli coefficient of a quantum state.

\begin{lemma}[Noise-induced exponential decay of Pauli coefficient~\cite{wang2020noise}]\label{exp-decay} Let $\bm{a}^{l}$ denote the Pauli coefficient of the quantum state $\rho_l$. Then $\|\bm{a}^{l}\|_\infty \leq \|\bm{a}^{l}\|_2 \le q^{l}\sqrt{2^n-1}$.
\end{lemma}

\begin{figure*}
     \centering
     \begin{subfigure}[b]{0.3\textwidth}
         \centering
         \includegraphics[width=\textwidth]{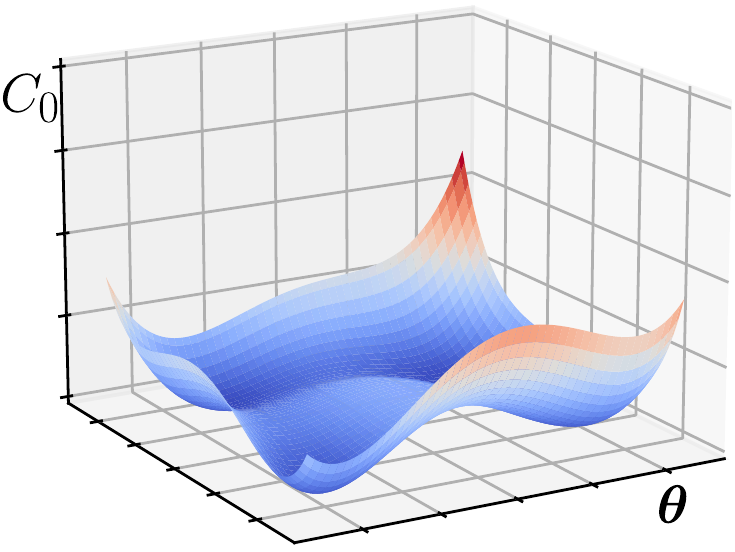}
         \vspace{-10pt}
         \caption{Cost function $C_0$}
         \label{fig:c0g}
     \end{subfigure}
     \hfill
     \begin{subfigure}[b]{0.3\textwidth}
         \centering
         \includegraphics[width=\textwidth]{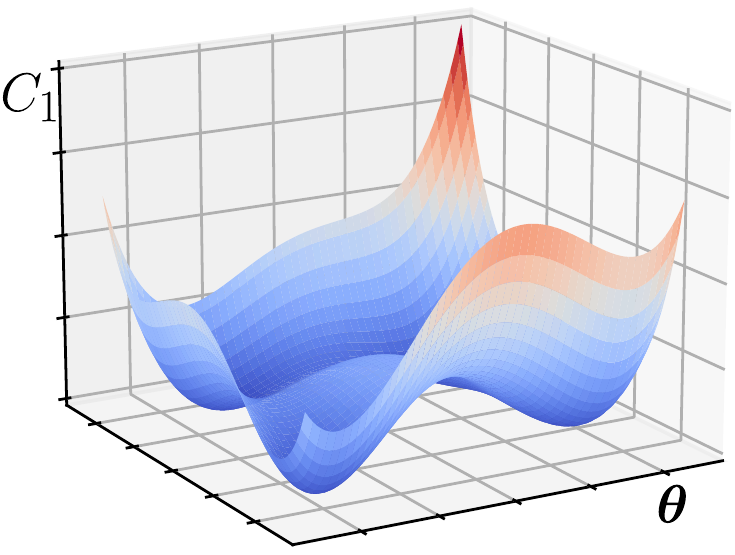}
         \vspace{-10pt}
         \caption{Cost function $C_1$}
         \label{fig:c1g}
     \end{subfigure}
     \hfill
     \begin{subfigure}[b]{0.3\textwidth}
         \centering
         \includegraphics[width=\textwidth]{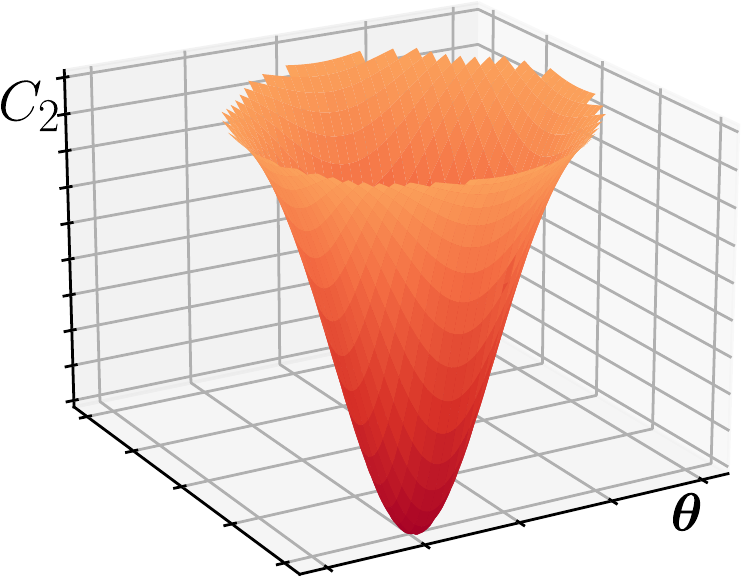}
         \vspace{-10pt}
         \caption{Cost function $C_2$}
         \label{fig:c2g}
     \end{subfigure}
     \vspace{-5pt}
        \caption{A schematic diagram for $C_0$, $C_1$, $C_2$. $C_2$ has larger gradient than $C_0$ and $C_2$ while it is singular.}
        \label{fig:three graphs}
        \vspace{-10pt}
\end{figure*}

\section{Augmented Gradient in Truncated Subspace}

In this section, we will show how to construct a cost function in a truncated subspace. 
We first define two new cost functions $C_1$ and $C_2$ in a truncated subspace.
The $C_1$ cost function will still have a small gradient but it is not singular near its minimum.
The $C_2$ cost function will have a large gradient and can be easily optimized with gradient descent but it is singular near its minimum, as shown in Figure~\ref{fig:three graphs}.
We will show that the minimum value point of $C_2$ is also a local minimum of $C_1$ and we can optimize $C_1$ with the gradient of $C_2$.

\subsection{Cost Function Definitions}
An observable (Hermitian operator) $O$ can always be decomposed into $O=\sum_{\lambda_i\ne 0}\lambda_i \qbit{i}{}$. 
The states in the summation terms span a subspace $S$ in the overall Hilbert space.
We first define a cost function $C_1$ in this subspace.

\begin{definition}[Cost function of truncated subspace]\label{def:cost1_original}
Given a subspace $S$, an observable $O$ and its eigendecomposition $O=\sum_{i}\lambda_i \qbit{i}{}$, the cost in this subspace is 
$C_1 =\frac{\sum_{i\in S}\lambda_i p_i }{ \sum_{i\in S} p_i}$,
where $p_i$ the measured probability of state $\ket{i}$. 
The denominator is for normalization.
\end{definition}

The cost function form in Definition \ref{def:cost1_original} is not suitable for gradient analysis.
We define an equivalent form of the cost function of truncated space in terms of observables. 

\begin{proposition}[Truncated cost function with observables]\label{prop:cost1_observable}
A cost function of truncated measurement  $C_1 =\frac{\sum_{i\in S}\lambda_i p_i }{ \sum_{i\in S} p_i}$ can also be expressed with two observables:
\begin{equation}\label{equ:defc1}
    C_1 = \frac{\Tr(\rho_L O_1)}{\Tr(\rho_L O_2)}\,,
\end{equation}
where $O_1 = \sum \limits_{i\in S} \lambda_i \vert i \rangle \langle i \vert$, $ O_2 =  \sum \limits_{i\in S} \vert i \rangle \langle i \vert$.
\end{proposition}
\myproof{Notice that $\Tr(\rho_L O_1) = \sum_{i\in S}\lambda_i p_i$, $\Tr(\rho_L O_1) = \sum_{i\in S}p_i$.}

\begin{example}[Measurement result truncation]
Suppose the observable is $O_1=0.8\qbit{00}{}+0.2\qbit{11}{}$. The quantum state produced by the variational quantum circuit is $\rho = 0.7\qbit{00}{}+0.26\qbit{11}{}+0.04\qbit{01}{}$. Measuring this state may give the result `01' with a probability $0.04$. We can discard this result since it is not in the subspace defined by $O_1$ and will not contribute to the cost function.

\end{example}

The Pauli representations of $O_1$ and $O_2$ are very useful for gradient analysis. We introduce them in the following:
$
    O_1 = \frac{\Tr(O_1)}{2^n}\id^{\otimes n} + \bm{w}_1\cdot \bm{\sigma}_n$, $O_2 = \frac{\Tr(O_2)}{2^n}\id^{\otimes n} + \bm{w}_2\cdot \bm{\sigma}_n
$.

Similar to the cost function $C_0$, the gradient of cost function $C_1$ decays exponentially as the number of layers increases, as depicted in Proposition~\ref{prop:gradc1}, because $C_1$ will converge to $\frac{\Tr(O_1)}{\Tr(O_2)}$ as the number of layers $L$ increases~\cite{wang2020noise}. 

\begin{proposition}[Gradient of $C_1$]\label{prop:gradc1}
$\forall \rho_L$ s.t. $\abs{\Tr(O_2\rho_L)} \ge  \abs{\frac{\Tr(O_2)}{2^n}}$, the gradient of $C_1$ will decay exponentially w.r.t the circuit depth $L$ as,
$\partial_{\theta_{lm}} C_1 \le \sqrt{2^{n+1}-2}h q^{L+1}$, 
where $h = (\frac{\Vert \bm{w}_1 \Vert_2}{k_2} + \max\limits_{i \in S} \lambda_i \frac{\Vert \bm{w}_2 \Vert_2}{k_2^2})N_{lm} \Vert \bm{\eta}_{lm} \Vert_\infty$, $k_2 = \frac{\Tr(O_2)}{2^n}$, $\bm{\eta}_{lm}$ and $N_{lm}$ are previously defined in Lemma~\ref{exp-decay}.
\end{proposition}
\myproof{Postponed to Appendix~\ref{apx:c1}.}

We observe that gradient vanishing is caused by the first term of the Pauli representation of the observables in the cost function. We quantitatively evaluate this effect in the following proposition.

\begin{proposition}[Dominating term in the observable]\label{prop:dom_term}
$$C_1 = \frac{ \frac{\Tr(O_1)}{2^n} + \bm{a}^L\cdot \bm{w}_1}{\frac{\Tr(O_2)}{2^n} + \bm{a}^L\cdot \bm{w}_2}\,,$$
where $\bm{a}^L$ is the Pauli coefficient of $\rho_L$.
As $L \to \infty$, if $\bm{a}^L \to 0$ and $\frac{\partial \bm{a}^L}{\partial \theta} \to 0$ , then $\partial_{\theta}C_1 \to 0$.
\end{proposition}
\myproof{Postponed to Appendix~\ref{apx:c1}.}

The Pauli representations of $O_1$ and $O_2$ 
have dominated term $\frac{\Tr(O_1)}{2^n}\id^{\otimes n}$ and $\frac{\Tr(O_2)}{2^n}\id^{\otimes n}$, respectively. These two dominated terms will mask the effect of other Pauli terms, making gradient vanishing quickly. 
However, by eliminating the dominated term, we can amplify the effect of other Pauli terms and thus mitigate the gradient vanishing. 
Therefore, we define a new cost function of truncated space with dominated terms eliminated.

\begin{definition}[Truncated cost function with traceless observables]\label{}
A cost function of truncated measurement with traceless observables is defined as
$C_2 = \frac{\Tr(\rho_L O_1')}{\Tr(\rho_L O_2')}$, 
where $O'_1 = O_1 - \frac{\Tr(O_1)}{2^n}\id^{\otimes n}$, $O'_2 = O_2 - \frac{\Tr(O_2)}{2^n}\id^{\otimes n}$.

\end{definition}

\subsection{Augmented Gradient}

After removing the first term in the Pauli representation of the observable, the new observables ($O'_1$, $O'_{2}$) become traceless ($\Tr(O'_1) = \Tr(O'_2) = 0$).
We first give the gradient of this new cost function with traceless observables in the following proposition:

\begin{proposition}[Gradient of $C_2$]\label{truncate gradient} The gradient of truncated cost function with traceless observables is
\begin{align}\label{equ:c2grad}
  \partial_\theta C_2 = \frac{\bm{g}^L\cdot \bm{w_1}}{\bm{a}^L\cdot \bm{w_2}} - \frac{\Tr(\rho_LO_1')}{\Tr(\rho_L O_2')}\frac{\bm{g}^L\cdot \bm{w_2}}{\bm{a}^L\cdot \bm{w_2}}\,,  
\end{align}
where $\bm{g}^L$ is the Pauli coefficient of $\partial_\theta \rho_L$, $\bm{a}^L$ is the Pauli coefficient of $\rho_L$.
\end{proposition}
\myproof{Postponed to Appendix~\ref{apx:c2}.}

Then we compare the gradient of the original cost function $C_1$ and that of the new cost function $C_2$ with traceless observables. 
In the following theorem, we will show that the gradient of $C_2$ is much larger than that of $C_1$ .

\begin{theorem}[Mitigating gradient vanishing]\label{theo:large_gradient}
Assume $L_0$ is the minimal integer s.t. $q^{L_0}\sqrt{2^n-1}\norm{ \bm{w}_2}_2(1+\epsilon) \le \abs{\frac{\Tr(O_2)}{2^n}}$, then for $L > 2L_0 + 1$, we have $\abs{\partial_{\theta_{lm}} C_2} \ge \frac{s}{q^{L+1}}\abs{\partial_{\theta_{lm}} C_1} - 1$,
where $s = \min{ \{\frac{\epsilon^2k_2^2}{(1+\epsilon)^2\sqrt{2^{n+1}-2}N_{lm}\norm{\bm{\eta}_{lm}}_\infty\norm{\bm{w}_1k_2 - \bm{w}_2k_1)}_2} ,} (\frac{\abs{ k_2}}{q^{L_0}\sqrt{2^n-1}\norm{\bm{w}_2}_2} - q^{L_0+1})^2 \}$, $\epsilon$ is an arbitrarily small positive number, $k_1 = \frac{\Tr(O_1)}{2^n}$, $k_2 = \frac{\Tr(O_2)}{2^n}$, $\bm{\eta}_{lm}$ and $N_{lm}$ are previously defined in Lemma~\ref{exp-decay}.
\end{theorem}
\myproof{Postponed to Appendix~\ref{apx:c2}.}

\subsection{Bypassing the Singularity}

However, the cost function $C_2$ has a major drawback that
it is singular around its minimum value point. 
We show its singularity in Proposition~\ref{prop:cost2_singular}.

\begin{proposition}[Singularity of $C_2$]\label{prop:cost2_singular} If $\max\limits_{i\in S} p_i = \sum\limits_{i\in S} p_i = \frac{\Tr(O_2)}{2^n}$, the cost function $C_2$ is singular and has the two following properties: (a) The global minimum of $C_2$ is negative infinity; (b) The gradient of $C_2$ will become infinitely large.
\end{proposition}
\myproof{Postponed to Appendix~\ref{apx:su_c1}.}

We may not want to directly optimize $C_2$ due to its singularity but we will show that we can use the gradient of $C_2$ to optimize $C_1$. 
In the following theorem, we reveal the connection between the minimum value points of $C_1$ and $C_2$. 
That is, the minimum value point of $C_2$ is also a minimum value point of $C_1$.

\begin{theorem}[Optimizing $C_1$ with gradient from $C_2$]\label{theo:optimization_scheme}
Assume the minimum value of $C_1$ be the $C_1^*$. %
If there is an iterative algorithm $\mathcal{P}$ that optimize $C_2$ with a sequence of parameters $\{ \bm{\theta}_1, \bm{\theta}_2, \cdots \}$, then
$\forall \epsilon > 0$,  $\exists n_0 \in \mathbb{N}$, s.t. $\forall n > n_0 $, $\vert C_1(\bm{\theta}_n) - C_1^* \vert < \epsilon$
where $\bm{\theta}_n$ is the parameter at the $n$-th iteration of $\mathcal{P}$.
\end{theorem}
\myproof{Postponed to Appendix~\ref{apx:su_c1}.}

Another benefit of optimizing $C_1$ rather than $C_2$ is that, the optimization process may terminate earlier for the cost function $C_1$ because $C_1$ has a large minimum value point space: %

\begin{proposition}[Solution space of $C_1$]\label{prop:cost1_soldim}
The solution space of problem $\min\limits_{\ket{\psi}} C_1$ is of dimension $2^n - \dim S + 1$.
\end{proposition}
\myproof{Postponed to Appendix~\ref{apx:su_c1}.}

Without loss of generality, we assume the minimum eigenvalue %
in the subspace $S$ is negative. In the following, we introduce a refined version of $C_2$ that has better properties in practical optimization.

\textbf{Avoiding singularity of $C_2$}:
To avoid the singularity of $C_2$'s gradient, we will add a threshold value to the denominator of $C_2$ as, $C_2 = \frac{\Tr (\rho_L O_1')}{ \Tr(\rho_L O_2') + \alpha}$,
where $\alpha$ is an adjustable positive hyper-parameter. %

\textbf{Penalize $C_2$'s gradient}:
To avoid leaving the subspace $S$, we will add a penalty term to the numerator of $C_2$ as, $C_2 = \frac{\Tr (\rho_L O_1') + \beta}{ \Tr(\rho_L O_2') + \alpha}$,
where $\beta$ is an adjustable positive hyper-parameter.

\section{Evaluation}

In this section, we evaluate the proposed gradient mitigation scheme over various VQA tasks. 

\subsection{Experimental Setup}

\textbf{Benchmarks:}
We select six different VQA benchmarks of two major types. 
The first three benchmarks solve three different combinatorial optimization problems, the graph max-cut problem (MC), the vertex cover problem (VC), and the traveling salesman problem (TSP), using the QAOA algorithm~\cite{farhi2014quantum}. 
The second three benchmarks simulate the ground state of three chemical molecules, ${\rm H_2}$, ${\rm LiH}$, and ${\rm BeH_2}$, using the VQE algorithm~\cite{peruzzo2014variational}.
Both QAOA and VQE are essentially VQA.
The information (numbers of qubits and parameters) of the circuits 
in these benchmarks are listed in Table~\ref{tab:benchmark}.
We illustrate the graphs of the three combinatorial optimization problems and the molecule configurations (atoms and bond lengths) in the `Description' column.

\begin{figure*}[t]
     \centering
     \includegraphics[width=\textwidth]{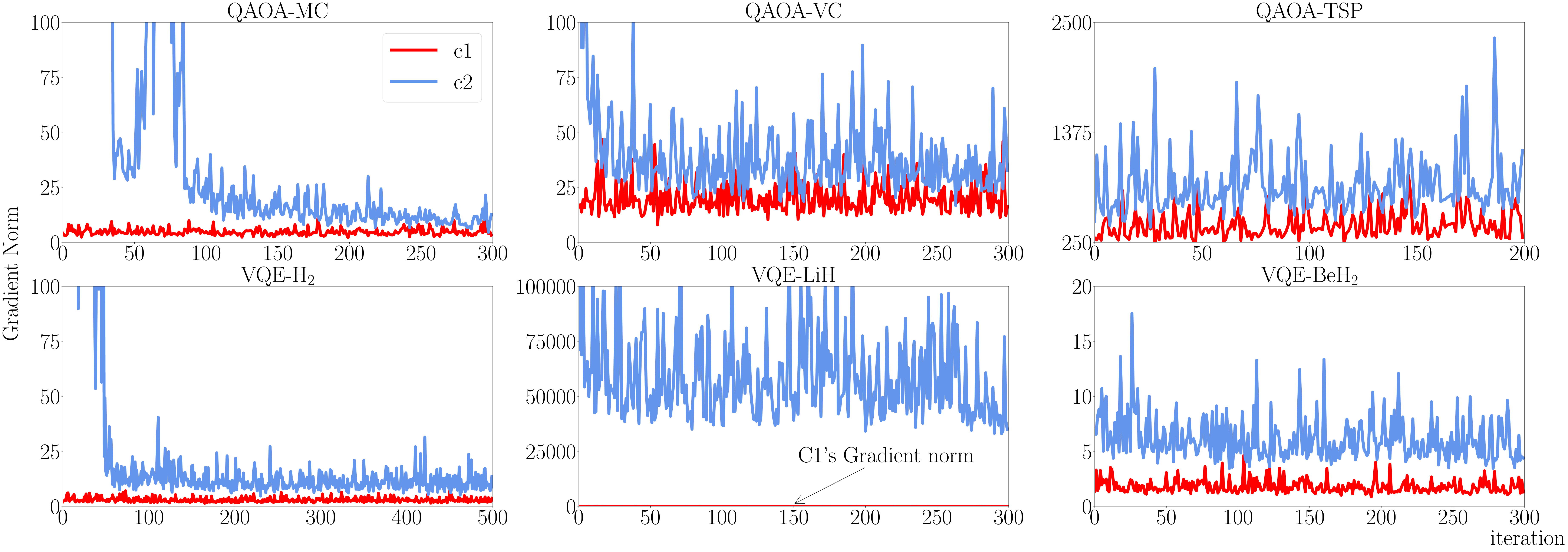}
     \vspace{-20pt}
     \caption{Gradient norm of $C_1$ and $C_2$ on QAOA and VQE tasks}
     \vspace{-10pt}
     \label{fig:gradient_compare}
\end{figure*}

\newcommand{\qstate}[1]{\ensuremath{|#1\rangle}}

The symmetry comes from the domain knowledge of the real VQA applications and it is widespread. For VQE application, such symmetry comes from the property of the simulated physical system governed by the physics laws. For example, suppose we are simulating a chemical system with 2 electrons on 3 orbitals. We will have 3 qubits representing the 3 orbitals and the qubit state 1 means the orbital is occupied by an electron. A subspace S could be $span\{\qstate{110}, \qstate{101}, \qstate{011}\}$ (the space spanned by basis states with two 1’s in the bit string). This is because the number of electrons is preserved in chemical simulation and 1 electron state (e.g., \qstate{011}) or 3 electron state (\qstate{111}) are not valid states. For QAOA, the symmetry can come from the underlying graph structure of the problem as well as the qubit encoding scheme of the problem. We recommend a good reference “Classical symmetries and qaoa” (https://arxiv.org/abs/2012.04713 (https://arxiv.org/abs/2012.04713)) which discusses the symmetries for QAOA. We have explained more details in Section 4.1 about how to find the symmetries in the benchmarks.

\textbf{Encoding and symmetries}: 
For VQE tasks, we adopt the commonly used Jordan–Wigner transformation~\cite{JordnberDP} and different basis states represent different numbers of electrons.
The symmetries of VQE tasks come from the fact that the number of electrons is conserved in chemical simulation.
The subspace is spanned by basis states corresponding to the correct number of electrons simulated.
The symmetries of QAOA tasks come from the graph structure.
For MC, we use 4 qubits $\{q_3,q_2,q_1,q_0\}$ to represent each graph node, and quantum state $\ket{1}$ ($\ket{0}$) means this node is selected in the first partition (resp. the second partition). Since it does not matter if the orders of the two partitions are exchanged, we can always assume $q_0 = \ket{1}$. 
For VC, the encoding scheme is similar to that of MC. 
We use $q_2$ to represent the rightmost node, $q_3$ to represent the uppermost node. Since $q_2$ and $q_3$ are nodes of reflection symmetry, we do not to need have both qubits being $\ket{1}$ for the VC task. Thus, we can always assume $q_2 = 0$. 
For TSP, we use 6 qubits $\{q_5,q_4,q_3,q_2,q_1,q_0\}$ to encode node numbers, i.e., 2 qubits per node. The graph of TSP has a cyclic symmetry, thus we can always assume $q_1q_0 = \ket{00}$. Our subspaces are derived from these symmetries.

\begin{table}[t]\small
  \centering
  \caption{Benchmarks}
    \resizebox{1\columnwidth}{!}{
\begin{tabular}{|c|c|c|c|c|c|c|}
\hline
Name         & QAOA-MC & QAOA-VC & QAOA-TSP & VQE-${\rm H_2}$ & VQE-${\rm LiH}$ & VQE-${\rm BeH_2}$ \\ \hline
\# of Qubits & 4       & 4       & 6        & 4               & 6               & 8                 \\ \hline
\# of Param. & 20      & 20      & 30       & 20              & 30              & 40                \\ \hline
Description &
  \cincludegraphics[height=8pt]{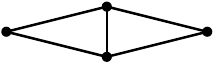} &
  \cincludegraphics[height=8pt]{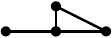} &
  \cincludegraphics[height=8pt]{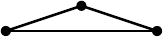} &
  \cincludegraphics[height=7pt]{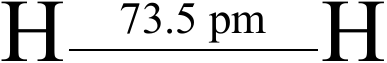} &
  \cincludegraphics[height=7pt]{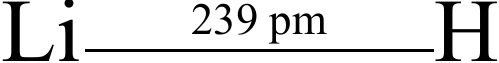} &
  \cincludegraphics[height=7pt]{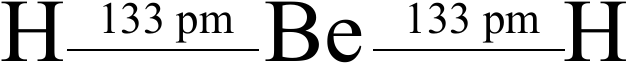} \\ \hline
\end{tabular}
    }
    \vspace{-15pt}
  \label{tab:benchmark}%
\end{table}%

\textbf{Gradient estimation:} The eigendecomposition of an observable, which exists but usually not available, is used only in our theoretical analysis.
When applying our technique in practice, we can obtain $\Tr (\rho_L O_1)$ via measurement. $\Tr (\rho_L O_2)$ can be estimated by accumulating the probability of states satisfying certain symmetries associated with the adopted subspace. Thus, we can compute $C_1$ or $C_2$ without eigendecomposition in evaluation. The gradients can then be calculated with the central difference method.

\textbf{Metrics:} We evaluate the values of the cost functions and the norms of the gradients during the training procedures to show the effect of our gradient vanishing mitigation technique. For QAOA tasks, we also include the \textit{Success Rate}, which is the probability of obtaining the correct answer in the final measurement. 
A higher successful rate is more desirable.
For VQE tasks, we use \textit{Fidelity} to measure the ``closeness'' of the ground state found in our VQE training and the true ground state of the target molecule. 
The fidelity $F$ of quantum states $\ket{\psi_1}$, $\ket{\psi_2}$ is defined by $F(\ket{\psi_1}, \ket{\psi_2}) = \big\vert \bra{\psi_1}\ket{\psi_2} \big\vert^{2}$ and $F(\ket{\psi_1}, \ket{\psi_2})=1 \Leftrightarrow \ket{\psi_1} = \ket{\psi_2}$.
When $F$ is closer to 1, the state found in VQA is closer to the true ground state.

To understand the quality of the learned parameters, we propose another metric called \textit{Parameter Quality}. 
For QAOA tasks, the parameter quality refers to the success rate of noise-free simulation using the parameters obtained from noisy simulation. 
For VQE tasks, the parameter quality refers to the fidelity between the state generated from noise-free simulation with the trained parameters and the correct eigenstate of the minimal eigenvalue. %

\textbf{Ansatz selection}: Our work is ansatz-independent since our propositions and theorems do not make any assumptions about the ansatz form. 
In this paper, we use hardware efficient ansatz (HEA) constructed by Qiskit for all experiments since HEA is the most widely used ansatz form in experiments due to its relatively small computation resource requirement.

\textbf{Noise setting}: In this paper, we only consider local Pauli noise.
For 4-qubit benchmarks, the noise rate is set to be $q_X=q_Y=q_Z = 0.03$. 
For 6-qubit and 8-qubit benchmarks, the noise rate is set to be $q_X=q_Y=q_Z = 0.01$. With this noise setting, our results on 10-20 circuit depth (later in this session) can already demonstrate the advantage of our method. 
We expect that our method can be more effective on larger-size ansatzes
according to Theorem~\ref{theo:large_gradient}.

\textbf{Implementation:} All cost functions and experiments are implemented on IBM's Qiskit framework~\cite{Qiskit}. For VQE experiments, we use Qiskit Chemistry and PySCF~\cite{PYSCF} to generate the Hamiltonian of the simulated molecules. For QAOA experiments, we use Qiskit Optimization to generate the Hamiltonian of the simulated combinatoric optimization problems. 
The parameters are optimized with the gradient descent algorithm. We perform all simulations based on IBM's Qiskit noisy simulation infrastructure (version 0.7.2). We perform all experiments on a server with a 6-core Intel E5-2603v4 CPU and 32GB of RAM. We use a linear decay for the learning rate.
The gradient is truncated if its $\infty$-norm is larger than certain threshold value.

\subsection{Results}

\begin{figure*}[t]
     \centering
     \includegraphics[width=\textwidth]{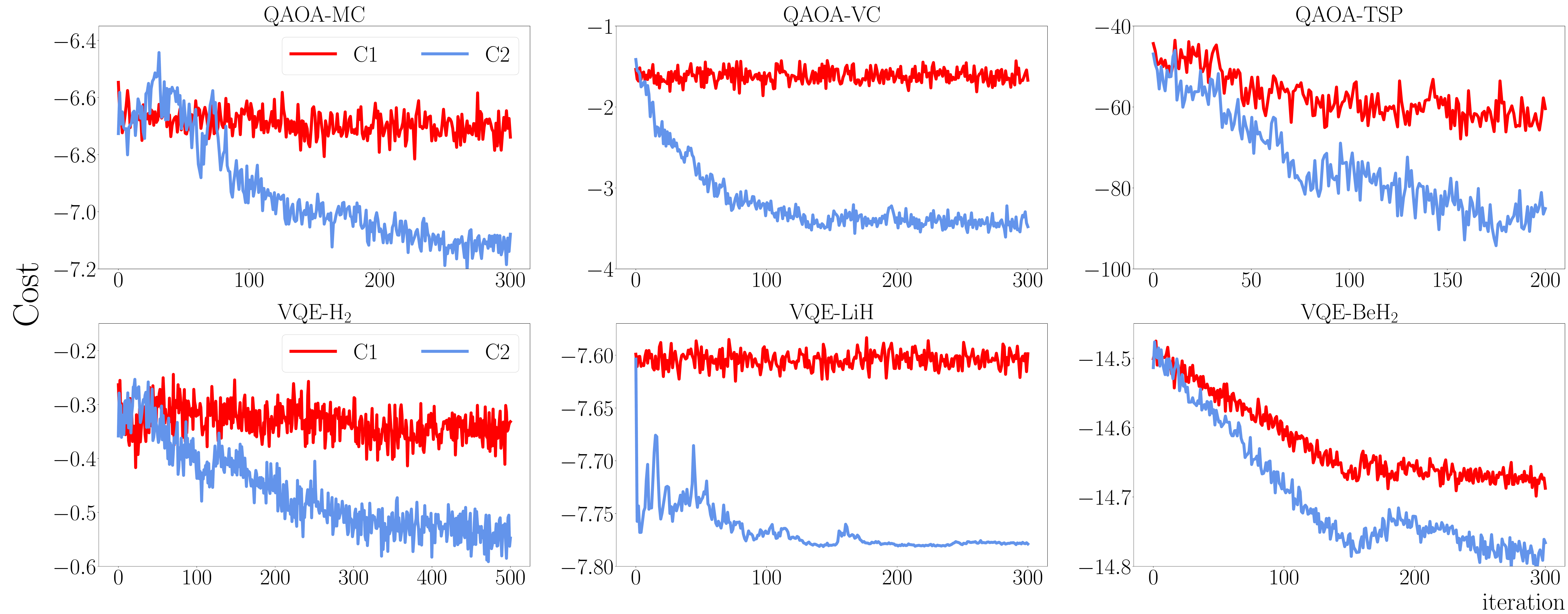}
     \vspace{-20pt}
     \caption{The optimization results with $C_1$ and $C_2$}
     \vspace{-10pt}
     \label{fig:cost}
\end{figure*}

\textbf{Augmented gradient}: 
We first compare the norm of the gradients from both $C_1$ and $C_2$.
As shown in Figure~\ref{fig:gradient_compare}, $C_2$ has a much larger gradient in nature  compared with $C_1$ under the same noise level and this result is consistent with Theorem~\ref{theo:large_gradient}.
The gradient norms of $C_2$ are about 141x, 5.7x, and 2.1x higher on average than that of $C_1$ for QAOA tasks MAXCUT, VC, and TSP, respectively. For VQE tasks,  the gradient norms of $C_2$ are about 21.7x, 2462.3x, and 3.7x higher on average than that of $C_1$ for the simulation of $\rm H_2$, $\rm LiH$, and $\rm BeH_2$, respectively.
The singularity of $C_2$ can also be observed since the gradient of $C_2$ can be very large (as discussed in Proposition~\ref{prop:cost2_singular}).

\begin{table}[t]
  \centering 
  \caption{Results for QAOA and VQE tasks.}
    \resizebox{\columnwidth}{!}{
    \begin{tabular}{|c|c|c|c|c|c|c|c|c|c|}
\hline
Benchmark &
  \multicolumn{2}{c|}{Success rate} &
  \multicolumn{2}{c|}{Parameter Quality} &
  Benchmark &
  \multicolumn{2}{c|}{Fidelity} &
  \multicolumn{2}{c|}{Parameter Quality} \\ \hline
Name     & $C_1$  & $C_2$  & $C_1$  & $C_2$  & Name              & $C_1$  & $C_2$  & $C_1$  & $C_2$  \\ \hline
QAOA-MC  & 0.3544 & 0.5380 & 0.4480 & 0.9156 & VQE-${\rm H_2}$   & 0.2100 & 0.4503 & 0.0242 & 0.8753 \\ \hline
QAOA-VC  & 0.1587 & 0.4757 & 0.1203 & 0.9570 & VQE-${\rm LiH}$   & 0.3980 & 0.8568 & 0.3633 & 0.8705 \\ \hline
QAOA-TSP & 0.2252 & 0.3718 & 0.6156 & 0.8974 & VQE-${\rm BeH_2}$ & 0.1681 & 0.2083 & 0.4108 & 0.5457 \\ \hline
\end{tabular}
    }
   \vspace{-15pt}
  \label{tab:qaoares}
\end{table}

\textbf{Overall improvement:} 
We then compare the overall optimization results of using $C_1$ and $C_2$.
In general, optimizations using the gradient from $C_2$ can converge to a better minimum. 
The upper half of Figure~\ref{fig:cost} shows the optimization results with cost function $C_1$ and $C_2$ on different QAOA tasks. 
Compared with using the gradient from $C_1$, $C_2$'s gradient increases the success rates significantly as shown in Table~\ref{tab:qaoares}: from 35.4\% to 53.8\% for the MC task, and from 15.9\% to 47.6\% for the VC task, from 22.5\% to 37.2\% for the TSP task, respectively. 
The parameter quality of using $C_2$ is also improved significantly as shown in Table~\ref{tab:qaoares}: from 44.8\% to 91.6\% for the MC task, and from 12.03\% to 95.7\% for the VC task, from 61.6\% to 89.7\% for the TSP task.

The lower half of Figure~\ref{fig:cost} shows the optimization results with cost function $C_1$ and $C_2$ on different VQE tasks and it can be observed that using gradients $C_2$ can converge to a better ground state. 
The final state fidelity improvement of using $C_2$ is shown in Table~\ref{tab:qaoares}: from 0.21 to 0.45 for the $\rm H_2$ molecule, from 0.398 to 0.857 for the $\rm LiH$ molecule, from 0.168 to 0.208 for the $\rm BeH_2$ molecule.
$C_2$ promotes the parameter quality significantly (Table~\ref{tab:qaoares}): from 0.024 to 0.875 for the $\rm H_2$ molecule, from 0.363 to 0.870 for the $\rm LiH$ molecule, from 0.411 to 0.546 for the $\rm BeH_2$ molecule.

In summary, VQAs with cost function $C_2$ show significant improvement over original VQAs using the vanilla cost $C_1$ and our optimization scheme in Theorem~\ref{theo:optimization_scheme} is very effective.
Sometimes using the gradient from $C_1$ cannot even optimize (e.g., QAOA-VC, VQE-LiH) because the gradient is mainly caused by the noise, and the gradient descent direction is random.
However, the gradient from $C_2$ is significantly augmented and is able to point out the correct descent direction.
This can help VQA converge faster and mitigate part of the noise effects.

\section{Discussion}

This paper studies mitigating the noise-induced gradient vanishing in VQA training through a novel training scheme. 
By identifying and eliminating the dominant term in the Pauli representation of an observable, we construct a new cost function that has significantly larger gradients without changing the minimum found in the gradient descent.
The advantages of our training approach are not only proved in theory but also demonstrated in practice.
The techniques proposed in this paper can mitigate one major obstacle in VQA training of larger sizes and thus help with a future demonstration of quantum advantage using VQAs.
Several future research directions are briefly discussed as follows:

\textbf{Beyond the cost function:} 
In this paper, we aim to mitigate the noise effects on the gradient and augment the gradient norm significantly through a new cost function $C_2$.
A VQA has several components other than the cost function, e.g., ansatz architecture, optimizer algorithm, parameter initialization.
It is worth exploring how we can improve these components to reduce the noise-induced gradient vanishing.

\textbf{More about the subspace:}
We present a general result on the benefit of using subspace. However, 
one VQA task may have multiple symmetries, leading to completely different subspace selections.
For example,  there are two symmetries in the TSP task. 
One symmetry is the so-called cyclic symmetry as any node can be the starting point.
Another symmetry is that the same node encoding cannot be repeated.
The subspace selection principles can be investigated to maximize the benefit with respect to the VQA performance.

\section{Conclusion}

This paper presents a novel training scheme to mitigate the inevitable noise-induced gradient vanishing in VQA training.
We show that, by introducing a new cost function with a traceless observable in a truncated subspace, the gradient can be augmented significantly.
We also prove that the same minimum value point can be found when optimizing the vanilla cost function with the augmented gradients from the new cost.
Experiments show that our new training scheme is applicable to major VQAs with various tasks.
The techniques proposed in this paper can alleviate the gradient vanishing obstacle on the road to demonstrating quantum advantages on near-term quantum computers.

\bibliographystyle{unsrt}
\bibliography{gradient_mitigation}

\appendix

\section{Gradient Vanishing of C1}\label{apx:c1}

In the following, we will follow the notations defined in previous sections.

We first prove that the derivative of a quantum state  with respect to the parameters in a variational quantum circuit is a traceless operator under the local Pauli channel noise model.

\begin{lemma}[Traceless derivative of the quantum state]\label{lemma:tl_dr}
For a variational circuit with local Pauli channel noise model, $\partial_\theta \rho_L$ is a traceless operator, i.e., $\Tr(\partial_\theta \rho_L) = 0$.
\end{lemma}
\begin{proof}
Assume $\bm{a}^L$ is the Pauli coefficient of $\rho_L$.
Since $\rho_L = \frac{1}{2^n}(\id^{\otimes n} + \bm{a}^L\cdot \bm{\sigma}_n$), we have
\begin{align}
    \partial_\theta \rho_L = \frac{1}{2^n}\frac{\partial \bm{a}^L}{\partial \theta} \cdot \bm{\sigma}_n
\end{align}
Thus, $\Tr(\partial_\theta \rho_L) = \Tr(\frac{1}{2^n}\frac{\partial \bm{a}^L}{\partial \theta} \cdot \bm{\sigma}_n) = 0$.
\end{proof}

Similar to the original cost function $C_0$, the gradient of cost function $C_1$ decays exponentially as the number of layers increases, as depicted in Proposition~\ref{prop:gradc1}, because $C_1$ will converge to $\frac{\Tr(O_1)}{\Tr(O_2)}$ as the number of layers $L$ increases~\cite{wang2020noise}. 

\textbf{Proposition 3.2} (Gradient of $C_1$)\textbf{.} \textit{$\forall \rho_L$ s.t. $\abs{\Tr(O_2\rho_L)} \ge  \abs{\frac{\Tr(O_2)}{2^n}}$, the gradient of $C_1$ will decay exponentially with respect to the circuit depth $L$: 
$$\partial_{\theta_{lm}} C_1 \le \sqrt{2^{n+1}-2}h q^{L+1}$$
where $h = (\frac{\Vert \bm{w}_1 \Vert_2}{k_2} + \max\limits_{i \in S} \lambda_i \frac{\Vert \bm{w}_2 \Vert_2}{k_2^2})N_{lm} \Vert \bm{\eta}_{lm} \Vert_\infty$, $k_2 = \frac{\Tr(O_2)}{2^n}$, $\bm{\eta}_{lm}$ and $N_{lm}$ are previously defined in Lemma~\ref{exp-decay}. }
\begin{proof}Let $\theta = \theta_{lm}$.
Using Lemma~\ref{lemma:tl_dr}, we have $\partial_\theta \rho_L = \frac{1}{2^n}\bm{g}^L\cdot \bm{\sigma}_n$, where $\bm{g}^L = \frac{\partial \bm{a}^L}{\partial \theta}$ is the Pauli coefficient of $\partial_\theta \rho_L$. Then
\begin{align*}
    \partial_{\theta} C_1 &= \frac{\Tr(O_1 \partial_{\theta} \rho_L)}{\Tr(O_2 \rho_L)} - \frac{\Tr(O_1 \rho_L)}{\Tr(O_2 \rho_L)}\frac{\Tr(O_2 \partial_{\theta} \rho_L)}{\Tr(O_2 \rho_L)} \\
    &= \frac{\Tr((\frac{\Tr(O_1)}{2^n}\id^{\otimes n}+\bm{w}_1\cdot \bm{\sigma}_n)(\frac{1}{2^n}\bm{g}^L\cdot \bm{\sigma}_n))}{\Tr(O_2 \rho_L)} - \frac{\Tr(O_1 \rho_L)}{\Tr(O_2 \rho_L)}\frac{\Tr((\frac{\Tr(O_2)}{2^n}\id^{\otimes n}+\bm{w}_2\cdot \bm{\sigma}_n)(\frac{1}{2^n}\bm{g}^L\cdot \bm{\sigma}_n))}{\Tr(O_2 \rho_L)}
    \\
    &= \frac{\bm{w}_1 \cdot \bm{g}^L}{\Tr(O_2 \rho_L)} - \frac{\Tr(O_1 \rho_L)}{\Tr(O_2 \rho_L)}\frac{\bm{w}_2 \cdot \bm{g}^L}{\Tr(O_2 \rho_L)}, 
\end{align*}
Then, we have
\begin{align*}
    \lvert \partial_{\theta} C_1 \rvert &\le \left\vert \frac{\bm{w}_1 \cdot \bm{g}^L}{\Tr(O_2 \rho_L)} \right\vert + \left\vert \frac{\Tr(O_1 \rho_L)}{\Tr(O_2 \rho_L)}\frac{\bm{w}_2 \cdot \bm{g}^L}{\Tr(O_2 \rho_L)} \right\vert \\
    &\le \vert \frac{\bm{w}_1 \cdot \bm{g}^L}{k_2} \vert + \max\limits_{i\in S} \lambda_i\vert \frac{1}{k_2}\frac{\bm{w}_2 \cdot \bm{g}^L}{k_2} \vert \\
    & \le \frac{1}{k_2}\norm{\bm{w}_1}_2 \norm{\bm{g}^L}_2 + \frac{1}{k_2^2}\max_{i\in S} \lambda_i \norm{\bm{w}_2}_2 \norm{\bm{g}^L}_2
\end{align*}
where $k_2 = \abs{\frac{\Tr(O_2)}{2^n}}$. %
With Lemma~\ref{lemma:grad_van_pauli}, we have
\begin{align*}
    \lvert \partial_{\theta} C_1 \rvert \le (\frac{1}{k_2}\norm{\bm{w}_1}_2 + \frac{1}{k_2^2}\max_{i\in S} \lambda_i \norm{\bm{w}_2}_2 )
    \sqrt{2^{n+1}-2}N_{lm}\norm{\bm{\eta}_{lm}}_\infty q^{L+1}
\end{align*}
Let $h = (\frac{\Vert \bm{w}_1 \Vert_2}{k_2} + \max\limits_{i \in S} \lambda_i \frac{\Vert \bm{w}_2 \Vert_2}{k_2^2})N_{lm} \Vert \bm{\eta}_{lm} \Vert_\infty$. Then we have $\partial_{\theta_{lm}} C_1 \le \sqrt{2^{n+1}-2}h q^{L+1}$.

\end{proof}

\textbf{Proposition 3.3} (Dominating Term in the Observable)\textbf{.}\textit{$$C_1 = \frac{ \frac{\Tr(O_1)}{2^n} + \bm{a}^L\cdot \bm{w}_1}{\frac{\Tr(O_2)}{2^n} + \bm{a}^L\cdot \bm{w}_2}$$
where $\bm{a}^L$ is the Pauli coefficient of $\rho_L$.
If $\bm{a}^L \to 0$ and $\frac{\partial \bm{a}^L}{\partial \theta} \to 0$ as $L \to \infty$, then}
\begin{align*}
    \operatornamewithlimits{lim}_{L\to\infty}\partial_{\theta}C_1 = 0
\end{align*}
\begin{proof}
$\Tr(\rho_L O_1) = \Tr(\frac{1}{2^n}(\id^{\otimes n} + \bm{a}^L\cdot \bm{\sigma}_n)(\frac{\Tr(O_1)}{2^n}\id^{\otimes n} + \bm{w}_1\cdot \bm{ \sigma}_n)) = \frac{\Tr(O_1)}{2^n} + \bm{a}^L\cdot \bm{w}_1$. Likewise, $\Tr(\rho_L O_2) = \frac{\Tr(O_2)}{2^n} + \bm{a}^L\cdot \bm{w}_2$. \\
We have
$$C_1 = \frac{\Tr(\rho_LO_1)}{\Tr(\rho_L O_2)} = \frac{ \frac{\Tr(O_1)}{2^n} + \bm{a}^L\cdot \bm{w}_1}{\frac{\Tr(O_2)}{2^n} + \bm{a}^L\cdot \bm{w}_2},$$
where $\bm{a}^L \to 0$ and $\frac{\partial \bm{a}^L}{\partial \theta} \to 0$ as $L \to \infty$ by Lemma~\ref{exp-decay} and Lemma~\ref{lemma:grad_van_pauli}. Thus,
$$\partial_\theta C_1 = \frac{(\frac{\partial \bm{a}^L}{\partial \theta}\cdot \bm{w}_1)(\frac{\Tr(O_2)}{2^n} + \bm{a}^L\cdot \bm{w}_2) - (\frac{\partial \bm{a}^L}{\partial \theta}\cdot \bm{w}_2)(\frac{\Tr(O_1)}{2^n} + \bm{a}^L\cdot \bm{w}_1)}{(\frac{\Tr(O_2)}{2^n} + \bm{a}^L\cdot \bm{w}_2)^2}\longrightarrow 0$$
as $L \to \infty$.

\end{proof}

\section{Gradient Vanishing Mitigation by C2}\label{apx:c2}

Observable $O_1$ has dominated term $\frac{\Tr(O_1)}{2^n}\id^{\otimes n}$ which will mask the effect of other Pauli terms, making gradient vanishing quickly. 
However, by eliminating the dominated term, we can amplify the effect of other Pauli terms and thus mitigate the gradient vanishing.

\textbf{Proposition 3.4} (Gradient of $C_2$)\textbf{.} \textit{The gradient of truncated cost function with traceless observables is}
\begin{align}
  \partial_\theta C_2 = \frac{\bm{g}^L\cdot \bm{w_1}}{\bm{a}^L\cdot \bm{w_2}} - \frac{\Tr(\rho_LO_1')}{\Tr(\rho_L O_2')}\frac{\bm{g}^L\cdot \bm{w_2}}{\bm{a}^L\cdot \bm{w_2}}  
\end{align}
\textit{where $\bm{g}^L$ is the Pauli coefficient of $\partial_\theta \rho_L$, $\bm{a}^L$ is the Pauli coefficient of $\rho_L$.}
\begin{proof}
With the derivative rule, we can get
$$\partial_\theta C_2 = \frac{\Tr(O_1' \partial_\theta \rho_L)}{\Tr(\rho_L O_2')} - \frac{\Tr(\rho_LO_1')}{\Tr(\rho_L O_2')}\frac{\Tr(O_2' \partial_\theta \rho_L)}{\Tr(\rho_L O_2')}$$
On the other hand, $\Tr(O_1' \partial_\theta \rho_L) = \Tr((\bm{w}_1\cdot\bm{\sigma}_n)\frac{1}{2^n}(\bm{g}^L\cdot\bm{\sigma}_n)) =  \bm{g}^L\cdot \bm{w}_1$. Likewise, $\Tr(O_2' \partial_\theta \rho_L) = \bm{g}^L\cdot \bm{w}_2$. And $\Tr(\rho_L O_2') =\Tr((\bm{w}_2\cdot\bm{\sigma}_n)\frac{1}{2^n}(\id^{\otimes n} + \bm{a}^L\cdot\bm{\sigma}_n)) = \bm{a}^L\cdot \bm{w}_2$. Substitute these equities into the above derivative expression, we can get the desired gradient expression.
\end{proof}

In the following theorem, we will show that the gradient of $C_2$ is much larger than that of $C_1$.

\textbf{Theorem 3.1} (Mitigating Gradient Vanishing)\textbf{.} \textit{Assume $L_0$ is the minimal integer s.t. $q^{L_0}\sqrt{2^n-1}\norm{ \bm{w}_2}_2(1+\epsilon) \le \abs{\frac{\Tr(O_2)}{2^n}}$, then for $L > 2L_0 + 1$, we have}
\begin{align*}
    \abs{\partial_{\theta_{lm}} C_2} \ge \frac{s}{q^{L+1}}\abs{\partial_{\theta_{lm}} C_1} - 1, 
\end{align*}
\textit{where $s = \min{ \{\frac{\epsilon^2k_2^2}{(1+\epsilon)^2\sqrt{2^{n+1}-2}N_{lm}\norm{\bm{\eta}_{lm}}_\infty\norm{\bm{w}_1k_2 - \bm{w}_2k_1)}_2} ,} (\frac{\abs{ k_2}}{q^{L_0}\sqrt{2^n-1}\norm{\bm{w}_2}_2} - q^{L_0+1})^2 \}$, $\epsilon$ is an arbitrarily small positive number, $k_1 = \frac{\Tr(O_1)}{2^n}$, $k_2 = \frac{\Tr(O_2)}{2^n}$, $\bm{\eta}_{lm}$ and $N_{lm}$ are previously defined in Lemma~\ref{exp-decay}.}
\begin{proof}
Let $\theta = \theta_{lm}$. From Proposition~\ref{truncate gradient}, we have:
\begin{align*}
    \partial_\theta C_2 &= \frac{\bm{g}^L\cdot \bm{w}_1}{\bm{a}^L\cdot \bm{w}_2} - \frac{\Tr(\rho_LO_1')}{\Tr(\rho_L O_2')}\frac{\bm{g}^L\cdot \bm{w}_2}{\bm{a}^L\cdot \bm{w}_2}
    \\ &= \frac{\bm{g}^L\cdot(\bm{w}_1(\bm{a}^L\cdot \bm{w}_2)-\bm{w}_2(\bm{a}^L\cdot \bm{w}_1))}{(\bm{a}^L\cdot \bm{w}_2)^2} 
    \\ &= \frac{\bm{g}^L\cdot(\bm{a}^L\times(\bm{w}_1\times \bm{w}_2))}{(\bm{a}^L\cdot \bm{w}_2)^2} 
\end{align*}

Let $k_1 = \frac{\Tr(O_1)}{2^n}$, $k_2 = \frac{\Tr(O_2)}{2^n}$, then
 
\begin{align*}
    \partial_\theta C_1 &= \frac{\bm{g}^L\cdot \bm{w}_1}{k_2 + \bm{a}^L\cdot \bm{w}_2} - \frac{\Tr(\rho_LO_1)}{\Tr(\rho_L O_2)}\frac{\bm{g}^L\cdot \bm{w}_2}{k_2+\bm{a}^L\cdot \bm{w}_2} 
    \\ &= \frac{\bm{g}^L\cdot(\bm{w}_1(k_2+\bm{a}^L\cdot \bm{w}_2)-\bm{w}_2(k_1+\bm{a}^L\cdot \bm{w}_1))}{(k_2+\bm{a}^L\cdot \bm{w}_2)^2}\\
    &= \frac{\bm{g}^L\cdot(\bm{a}^L\times(\bm{w}_1\times \bm{w}_2))}{(k_2+\bm{a}^L\cdot \bm{w}_2)^2} + \frac{\bm{g}^L\cdot(\bm{w}_1k_2 - \bm{w}_2k_1)}{(k_2+\bm{a}^L\cdot \bm{w}_2)^2} \\
    & = \partial_\theta C_2 \frac{(\bm{a}^L\cdot \bm{w}_2)^2}{(k_2+\bm{a}^L\cdot \bm{w}_2)^2} + \frac{\bm{g}^L\cdot(\bm{w}_1k_2 - \bm{w}_2k_1)}{(k_2+\bm{a}^L\cdot \bm{w}_2)^2}
\end{align*}
Let $h = \max\{\frac{(\bm{a}^L\cdot \bm{w}_2)^2}{(k_2+\bm{a}^L\cdot \bm{w}_2)^2},  \frac{\abs{\bm{g}^L\cdot(\bm{w}_1k_2 - \bm{w}_2k_1)}}{(k_2+\bm{a}^L\cdot \bm{w}_2)^2} \}$. Then,
\begin{align*}
    \abs{\partial_\theta C_1} \le \abs{\partial_\theta C_2}\frac{(\bm{a}^L\cdot \bm{w}_2)^2}{(k_2+\bm{a}^L\cdot \bm{w}_2)^2} + \frac{\abs{\bm{g}^L\cdot(\bm{w}_1k_2 - \bm{w}_2k_1)}}{(k_2+\bm{a}^L\cdot \bm{w}_2)^2} \le h(\abs{\partial_\theta C_2} + 1)
\end{align*}

By Lemma~\ref{exp-decay}, we have $\abs{\bm{a}^{L_0}\cdot \bm{w}_2}(1+\epsilon) \le q^{L_0}\sqrt{2^n-1}\norm{ \bm{w}_2}_2(1+\epsilon) \le \abs{k_2}$. In the following, we consider two cases: \\
Case (1): $(\bm{a}^L\cdot \bm{w}_2)^2 > \abs{\bm{g}^L\cdot(\bm{w}_1k_2 - \bm{w}_2k_1)}$. Then
\begin{align*}
    \frac{1}{h} &= \frac{(k_2+\bm{a}^L\cdot \bm{w}_2)^2}{(\bm{a}^L\cdot \bm{w}_2)^2} = (1+\frac{ k_2}{\bm{a}^L\cdot \bm{w}_2})^2 \\
    &\ge (\abs{\frac{ k_2}{\bm{a}^L\cdot \bm{w}_2}}-1)^2 
    \ge 
    (\frac{\abs{ k_2}}{q^L\sqrt{2^n-1}\norm{\bm{w}_2}_2}-1)^2 
    \\& \ge \frac{1}{q^{2L-2L_0}}(\frac{\abs{ k_2}}{q^{L_0}\sqrt{2^n-1}\norm{\bm{w}_2}_2} - q^{L-L_0})^2 \\
    & \ge \frac{1}{q^{L+1}}(\frac{\abs{ k_2}}{q^{L_0}\sqrt{2^n-1}\norm{\bm{w}_2}_2} - q^{L_0+1})^2
\end{align*}

Case (2): $(\bm{a}^L\cdot \bm{w}_2)^2 < \abs{\bm{g}^L\cdot(\bm{w}_1k_2 - \bm{w}_2k_1)}$. Then,
\begin{align*}
    \frac{1}{h} &= \frac{(k_2+\bm{a}^L\cdot \bm{w}_2)^2}{\abs{\bm{g}^L\cdot(\bm{w}_1k_2 - \bm{w}_2k_1)}} \ge \frac{(\abs{k_2}-\abs{\bm{a}^L\cdot \bm{w}_2})^2}{\abs{\bm{g}^L\cdot(\bm{w}_1k_2 - \bm{w}_2k_1)}} \\
    & \ge \frac{\epsilon^2k_2^2}{(1+\epsilon)^2\abs{\bm{g}^L\cdot(\bm{w}_1k_2 - \bm{w}_2k_1)}} \\&\ge \frac{1}{q^{L+1}} \frac{\epsilon^2k_2^2}{(1+\epsilon)^2\sqrt{2^{n+1}-2}N_{lm}\norm{\bm{\eta}_{lm}}_\infty\norm{\bm{w}_1k_2 - \bm{w}_2k_1)}_2}
\end{align*}

Let $s = \min \{ (\frac{\abs{ k_2}}{q^{L_0}\sqrt{2^n-1}\norm{\bm{w}_2}_2} - q^{L_0+1})^2, \frac{\epsilon^2k_2^2}{(1+\epsilon)^2\sqrt{2^{n+1}-2}N_{lm}\norm{\bm{\eta}_{lm}}_\infty\norm{\bm{w}_1k_2 - \bm{w}_2k_1)}_2} \}$, then we have 
\begin{align*}
    \abs{\partial_\theta C_2} + 1 &\ge \frac{1}{h}\abs{\partial_\theta C_1} 
    \ge \frac{s}{q^{L+1}} \abs{\partial_\theta C_1}
\end{align*}
\end{proof}

\section{Surrogate Gradient for C1}\label{apx:su_c1}

However, the cost function $C_2$ has a major drawback that
it is singular around its minimum value point. 
We show its singularity in Proposition~\ref{prop:cost2_singular}.

\textbf{Proposition 3.5} (Singularity of $C_2$)\textbf{.} \textit{If $\max\limits_{i\in S} p_i = \sum\limits_{i\in S} p_i = \frac{\Tr(O_2)}{2^n}$,} \textit{the cost function $C_2$ is singular and has the two following properties:}
\begin{itemize}
    \item \textit{The global minimum of $C_2$ is negative infinity}
    \item \textit{The gradient of $C_2$ will become infinitely large}
\end{itemize}
\begin{proof}
Let $g(\rho_L, O_1, O_2) = C_2 = \frac{\Tr(\rho_L O_1) - \frac{\Tr(O_1)}{2^n}}{\Tr(\rho_L O_2) - \frac{\Tr(O_2)}{2^n}}$, $k_1 = \frac{\Tr(O_1)}{2^n}$, $k_2 = \frac{\Tr(O_2)}{2^n} > 0$. \\
Notice that $g(\rho, kO_1, O_2)$ is linear w.r.t $O_1$, i.e. $g(\rho, kO_1, O_2) = kg(\rho, O_1, O_2)$, then we can rewrite $g(\rho, O_1, O_2)$: 
$$g(\rho, O_1, O_2) = \frac{1}{k_{21}}\frac{\Tr(\rho_L k_{21}O_1) - k_2}{\Tr(\rho_L O_2) - k_2}$$
where $k_{21} = \frac{k_2}{k_1}$. \\
To find the minimum of $g(\rho, O_1, O_2)$, we could adopt a two-dimensional method: \\
(1) For a given constant $k_2 < m $, $g_c(m, O_1, O_2) = \min g(\rho, O_1, O_2)$ s.t. $\sum_{j\in S} p_j = m$. \\
(2) Find $\min_m g_c(m, O_1, O_2)$. \\
For task (1), the minimum of $g(\rho, O_1, O_2) = \frac{\sum_{i\in S}p_i\lambda_i - k_1}{\sum_{i\in S}p_i - k_2}$ is achieved when $p_i = m$ where $i = \text{arg}\min\limits_{j \in S} \lambda_j$. 
Thus, $g_c$ consists of values $g_c(m, O_1, O_2) = \frac{1}{k_{21}}\frac{mk_{21}\lambda_i - k_2}{m - k_2} = \lambda_i + \frac{(\lambda_i - \frac{1}{k_{21}})k_2}{m - k_2}$. 
Since $\lambda_i < \frac{1}{k_{21}}$,
the minimum value of $g_c$ (task (2)) is negative infinity. And, the gradient of $g_c$ at point $m = k_2 + 0_+$ is $+\infty$.

\end{proof}

We may not want to directly optimize $C_2$ due to its singularity but we will show that we can use the gradient of $C_2$ to optimize $C_1$. 
In the following theorem, we reveal the connection between the minimum value points of $C_1$ and $C_2$. 
That is, each minimum values point of $C_2$ is also a minimum value point of $C_1$.

\textbf{Theorem 3.2} (Optimizing $C_1$ with gradient from $C_2$)\textbf{.} \textit{Assume the minimum value of $C_1$ be the $C_1^*$. If there is an iterative algorithm $\mathcal{P}$ that optimize $C_2$ with a sequence of parameters $\{ \bm{\theta}_1, \bm{\theta}_2, \cdots \}$, then
$\forall \epsilon > 0$,  $\exists n_0$, s.t. $\forall n > n_0$, $\vert C_1(\bm{\theta}_n) - C_1^* \vert < \epsilon$
where $\bm{\theta}_n$ is the parameter at the $n$-th iteration of $\mathcal{P}$.}
\begin{proof}
Assume $\bm{\theta}^*_1$ is the minimum point found by $\mathcal{P}$ with initial condition $\sum_{i\in S}p_i > \frac{\Tr(O_2)}{2^n}$. \\
Then, as indicated in the proof of the Proposition~\ref{prop:cost2_singular}, we have $p_i(\bm{\theta}^*_1) = \sum_{j\in S}p_j$ and $p_j(\bm{\theta}^*_1) = 0$ if $j\in S$ and $j\ne i$, where $i = \text{arg}\min\limits_{j\in S}\lambda_j$. Thus, $C_1(\bm{\theta}^*_1) = C_1^* = \min\limits_{j\in S}\lambda_j$.\\

Since $C_1$ is continuous, thus, $\forall \epsilon > 0$, $\exists \delta > 0$, s.t. $\vert C_1(\bm{\theta}_n) - C_1(\bm{\theta}^*_1) \vert = \vert C_1(\bm{\theta}_n) - C_1^*\vert < \epsilon$ as long as $\Vert \bm{\theta}_n - \bm{\theta}^*_1 \Vert < \delta$. \\
On the other hand, $C_2$ is optimized by $\mathcal{P}$ iteratively, thus $\forall \delta > 0$, with the algorithm $\mathcal{P}$, exists $n_0$, s.t. $\forall n > n_0$, $\Vert \bm{\theta}_n - \bm{\theta}^*_1 \Vert < \delta$. This completes the proof. \\
\end{proof}

Another benefit of optimizing $C_1$ rather than $C_2$ is that, the optimization process may terminated earlier cost function $C_1$ because $C_1$ has a large minimum value point space as revealed by Proposition~\ref{prop:cost1_soldim}. 

\textbf{Propositon 3.6} (Solution space of $C_1$)\textbf{.}
\textit{The solution space of problem $\min\limits_{\ket{\psi}} C_1$ is of dimension $2^n - \dim S + 1$.
}
\begin{proof}
We denote the minimum eigenvalue in subspace $S$ by $\lambda_{\min}^S$. Then
\begin{align}
    C_1 = \frac{\sum_{i\in S}p_i \lambda_i}{\sum_{i\in S}p_i} \ge \lambda_{\min}^S 
\end{align}
Let $i = \text{arg}\min\limits_{j \in S} \lambda_j$.The minimum value of $C_1$ is achieved only when $p_j = 0$, $\forall j \in S, j \ne i$. \\
Thus, any solution of the form $\sum_{j \in \bar{S}} a_j \vert j \rangle + p_i \vert i \rangle$ is the minimum point of $C_1$, as long as $p_i > 0$. \\
Thus, the solution space of problem $\min C_1$ is of dimension $2^n - \dim S + 1$. \\
\end{proof}

\end{document}